\documentclass[runningheads]{llncs}

\usepackage[a4paper, total={6in, 8in}]{geometry}

\usepackage{amsmath}
\usepackage{multirow}
\usepackage{graphicx}
\usepackage{amsfonts}

\newcommand{\cost}{\textnormal{cost}}
\newcommand{\A}{\mathcal{A}}
\newcommand{\N}{\mathcal{N}}

\newcommand{\E}{\mathbb{E}}
\newcommand{\Var}[1]{\textnormal{Var}(#1)}
\newcommand{\OPT}{\mathsf{OPT}}

\usepackage[ruled]{algorithm2e}

\title{Adaptivity Gaps for the
Stochastic Boolean Function Evaluation Problem} 

\author{Lisa Hellerstein \inst{1} \and
Devorah Kletenik \inst{2} \and
Naifeng Liu \inst{3} \and
R. Teal Witter \inst{1}}

\institute{NYU Tandon, Brooklyn NY 11201, USA \and
Brooklyn College, Brooklyn NY 11210, USA \and
CUNY Graduate Center, New York NY 10016, USA}

\begin{document}

\maketitle

\begin{abstract}
    We consider the Stochastic Boolean Function Evaluation (SBFE) problem where the task is to efficiently evaluate a known Boolean function $f$ on an unknown bit string $x$ of length $n$. We determine $f(x)$ by sequentially testing the variables of $x$, each of which is associated with a cost of testing and an independent probability of being true. If a strategy for solving the problem  is adaptive in the sense that its next test can depend on the outcomes of previous tests, it has lower expected cost but may take up to exponential space to store. In contrast, a non-adaptive strategy may have higher expected cost but can be stored in linear space and benefit from parallel resources. The adaptivity gap, the ratio between the expected cost of the optimal non-adaptive and adaptive strategies, is a measure of the benefit of adaptivity. We present lower bounds on the adaptivity gap for the SBFE problem for popular classes of Boolean functions, including read-once DNF formulas, read-once formulas, and general DNFs. Our bounds range from $\Omega(\log n)$ to $\Omega(n/\log n)$, contrasting with recent $O(1)$ gaps shown for symmetric functions and linear threshold functions.
\end{abstract}

\section{Introduction}

We consider the question of determining adaptivity gaps for the Stochastic Boolean Function Evaluation (SBFE) problem, for different classes of Boolean formulas.
In an SBFE problem,
we are given a (representation of a) Boolean function 
$f:\{0,1\}^n \rightarrow \{0,1\}$,
a positive cost vector $c=[c_1, \ldots, c_n]$,
and a probability vector $p=[p_1, \ldots, p_n]$.
The problem is to determine the value $f(x)$
on an initially unknown random input $x \in \{0,1\}^n$.
The value of each $x_i$ can only be determined by performing a test, which incurs a cost of $c_i$.  Each $x_i$ is equal to 1 (is true) with independent probability $p_i$.
Tests are performed sequentially and
continue until $f(x)$ can be determined.
We say $f(x)$ is determined by a set of tests
if $f(x) = f(x')$ for all $x' \in \{0,1\}^n$ 
such that $x'_i=x_i$ for every $i$ in the set of tests. 

For example, if $f(x) = x_1 \vee \ldots \vee x_n$  then testing continues until a test is performed on some $x_i$ such that $x_i=1$ at which point we know $f(x)=1$,
or until all $n$ tests have been performed 
with outcome $x_i=0$ for each $x_i$ so we know $f(x)=0$.
The problem is to determine the order to perform tests that minimizes
the total expected cost of the tests. 

We will call a testing order a strategy which we can think of
as a decision tree for evaluating $f$.
A strategy can be {\em adaptive},
meaning that the choice of the next test $x_i$ can depend 
on the outcome of previous tests. 
In some practical settings, however, it is desirable 
to consider only {\em non-adaptive} strategies.
Non-adaptive strategies often take up less space than adaptive strategies, and they may be able to be evaluated more quickly if tests can be performed in parallel~\cite{DBLP:conf/soda/GuptaNS17}, such as in the problem of detecting network faults~\cite{harvey2007non} or in group testing for viruses, such as the coronavirus~\cite{liva2021optimum}.
A non-adaptive testing strategy is a permutation of the tests
where
testing continues in the order specified by the permutation 
until the value of $f(x)$ can be determined from the outcomes of the tests performed so far.
A non-adaptive strategy also corresponds to a decision tree
where all non-leaf nodes on the same level contain the same test $x_i$.

The {\em adaptivity gap} measures how much benefit 
can be obtained by using an adaptive strategy.
Consider a class $F$ of $n$-variable functions $f:\{0,1\}^n \rightarrow \{0,1\}$.
Let $\OPT_\N(f,c,p)$ be the expected evaluation cost of the optimal
non-adaptive strategy on function $f$ under costs $c$ and probabilities $p$.
Similarly, $\OPT_\A(f,c,p)$ is the expected evaluation cost of
the optimal adaptive strategy on $f$ under $c$ and $p$.
The adaptivity gap of the function class $F$ is
\begin{align*}
    \max_{f \in F}
    \sup_{c,p}
    \frac{\OPT_\N(f,c,p)}{\OPT_\A(f,c,p)}.
\end{align*}

The SBFE problem for a class of Boolean 
formulas $F$ restricts the evaluated $f$ to be a member of $F$.
In this paper we prove bounds on the adaptivity gaps for the SBFE problem on 
read-once DNF formulas,
DNF formulas, and read-once formulas.
(See Section~\ref{preliminaries} for definitions.)
A summary of our results can be found in Table~\ref{table:results}.


All the 
bounds in the table have a dependence on $n$, meaning that none of the listed SBFE problems has a constant adaptivity gap.  This contrasts with recent work of Ghuge et al.~\cite{Ghuge21}, which shows that the adaptivity gaps for the SBFE problem for symmetric Boolean functions and linear threshold functions are $O(1)$. 

For any SBFE problem, the non-adaptive strategy of testing the $x_i$ in increasing order of $c_i$ has an expected cost that is within a factor of $n$ of the optimal adaptive strategy~\cite{DBLP:conf/stoc/KaplanKM05}.  Thus $n$ is an upper bound on the adaptivity gap for all SBFE problems.

\begin{table}[ht]
\centering
\caption{A summary of our results. We also prove an $O(\sqrt{n})$
upper bound for tribes formulas i.e., read-once DNFs with unit costs
where every term has the same number of variables.
We say all probabilities are equal if $p_1 = p_2 = \ldots = p_n$.}
\label{table:results}
{\renewcommand{\arraystretch}{1.2}
\begin{tabular}{|l|l|}
 \hline
 \textbf{Formula Class} &  \textbf{Adaptivity Gap} \\
\hline
\multirow{3}{*}{Read-once DNF}
&  $\Theta(\log n)$ for unit costs, uniform distribution\\
\cline{2-2}
 &  $\Omega(\sqrt{n})$ for unit costs\\
\cline{2-2}
& $\Omega(n^{1-\epsilon}/\log n)$ for uniform distribution\\
\hline
Read-once & $\Omega(\epsilon^3 n^{1-2\epsilon/\ln(2)})$ for unit costs, equal probabilities\\
\hline
\multirow{2}{*}{DNF}
& $\Omega(n/\log n)$ for unit costs, uniform distribution\\
\cline{2-2}
& $\Theta(n)$ for uniform distribution \\
\hline
\end{tabular}
}
\end{table}


\textbf{Outline:} We present our results on formula classes
in increasing order of generality.
In Section \ref{sec:read-once}, we warm up with a variety of
results on read-once DNF formulas in different settings.
In Section \ref{sec:formula},
we prove our main technical result for read-once formulas,
drawing on branching process identities and concentration 
inequalities.
In Section \ref{sec:dnf}, we prove our most general results
on DNF formulas (the bounds also apply to the restricted
class of DNF formulas with a linear number of terms).
Note that we state our lower bound results in the most 
restricted context because they of course apply to more
general settings.
Due to space constraints, we defer some proofs to 
Appendix $\ref{app:delayed}$.

\subsection{Connection to $st$-connectivity in uncertain networks}

Our result for read-once formulas has implications for a problem
of determining $st$-connectivity in an uncertain network,
studied by Fu et al.~\cite{DBLP:journals/ton/FuFXPWL17}.
The input is a multi-graph with a 
source node $s$ and a destination node $t$.
Each edge corresponds to a variable $x_i$ indicating whether it is usable,
which is true with probability $p_i$.
Testing the usability of edge $i$ costs $c_i$.
The $st$-connectivity function for the multi-graph is true if and only if
there is a path of usable edges from $s$ to $t$.
The problem is to find a strategy to evaluate
the $st$-connectivity function that has minimum expected cost.



The $st$-connectivity function associated with a multi-graph can be represented 
by a read-once formula if and only if the multi-graph is a two-terminal series-parallel graph.
This type of graph has two distinguished nodes, $s$ and $t$, and is formed by recursively combining disjoint series-parallel graphs either in series, or in parallel (see \cite{eppstein1992parallel} for the precise definitions).\footnote{The term {\it series-parallel circuits} (systems) refers to a set of parallel circuits that are connected in series (see, e.g.,~\cite{el1986optimal,wiki:Series_and_parallel_circuits}).
Viewed as graphs,
they correspond to the subset of two-terminal series-parallel graphs whose $st$-connectivity functions correspond to read-once CNF formulas.  We note that Kowshik used the term ``series-parallel graph'' in a non-standard way to refer only to this subset; Fu et al.\, in citing Kowshik, used the term the same way~\cite{kowshik2011information,DBLP:journals/ton/FuFXPWL17}.}
Fu et al.\ performed experiments with both adaptive and non-adaptive strategies for this problem, comparing their performance, but did not prove theoretical adaptivity gap bounds.
Since the $st$-connectivity function on a series-parallel graph
is a read-once formula,
our lower bound on the adaptivity gap for read-once formulas
applies to the problem of $st$-connectivity. 

\subsection{Related work}
It is well-known that the SBFE problem for the Boolean OR function given by $f(x) = x_1 \vee \ldots \vee x_n$ has a simple solution:
test the variables $x_i$ in increasing order of the ratio $c_i/p_i$ until  a test reveals a variable set to true, or until all variables are tested and found to be false (cf.~\cite{Unluyurt04}). This strategy is non-adaptive, meaning that the SBFE problem for the Boolean OR function has an adaptivity gap of 1. That is, there is no benefit to adaptivity. 

Gkenosis et al.~\cite{gkenosisetal18} introduced the Stochastic Score Classification problem, which generalizes the SBFE problem for both symmetric Boolean functions and for linear threshold functions.  
Ghuge et al.~\cite{Ghuge21} showed that the Stochastic Score Classification problem has an adaptivity gap of $O(1)$.
In the unit-cost case, Gkenosis et al.\ showed the gap is at most 4 for symmetric Boolean functions, and at most  $\phi$ (the golden ratio) for the not-all-equal function~\cite{gkenosisetal18}.

Adaptivity gaps were introduced by Dean et al.~\cite{DBLP:conf/focs/DeanGV04} in the study of the stochastic knapsack problem 
which, in contrast to the SBFE problem for Boolean functions, is a maximization problem.
It has an  adaptivity gap of 4  ~\cite{DBLP:conf/focs/DeanGV04,DBLP:journals/mor/DeanGV08}.
Adaptivity gaps have also been shown for other stochastic maximization problems  (e.g.,
\cite{DBLP:conf/soda/DeanGV05,DBLP:conf/soda/GuptaNS16,DBLP:conf/approx/Bradac0Z19,HKL15,DBLP:conf/wine/AsadpourNS08}).  
Notably, the problem of
maximizing a monotone submodular function with stochastic inputs, subject to any class of prefix-closed constraints,
 was shown to have an $O(1)$ adaptivity gap~\cite{DBLP:conf/soda/GuptaNS17,DBLP:conf/approx/Bradac0Z19}.

Adaptivity gaps have also been shown for stochastic covering problems, which, like SBFE, are minimization problems.
Goemans et al.~\cite{DBLP:conf/latin/GoemansV06} showed that
the adaptivity gap for the Stochastic Set Cover problem, in which each item can only be chosen once, is $\Omega(d)$ and $O(d^2)$,
where $d$ is the size of the target set to be covered.
If the items can be used repeatedly, the adaptivity gap is $\Theta(\log d)$. 


Agarwal et al.~\cite{DBLP:conf/soda/AgarwalAK19} 
and Ghuge et al.~\cite{GhugeetalStochasticSubmod} proved bounds of $\Omega(Q)$ and $O(Q \log Q)$ respectively,
on the adaptivity gap for the more abstract Stochastic Submodular Cover Problem in which each item can only be used once. 
Applied to the special case of Stochastic Set Cover, the upper bound is $O(d \log d)$, which improves the above $O(d^2)$ bound.
They also gave bounds parameterized by the number of rounds of adaptivity allowed.
We note that, as shown by Deshpande et al.~\cite{DeshpandeGoalValue}, one approach to solving SBFE problems is to reduce them to special cases of Stochastic Submodular Cover.  However, this approach does not seem to have interesting implications for SBFE problem adaptivity gaps.

\subsection{Preliminaries}\label{preliminaries}

Consider a Boolean function $f:\{0,1\}^n \rightarrow \{0,1\}$,
a positive cost vector $c=[c_1,\ldots,c_n]$, and probability vector
$p=[p_1,\ldots, p_n]$.
We assume $c_i > 0$ and $0 < p_i < 1$ for $i \in [n]$
where $[n]$ denotes the set $\{1, \ldots, n\}$.
Let strategy $S$ be a decision tree for evaluating $f(x)$
on an unknown input $x \in \{0,1\}^n$.
We define $\cost_c(f, x, S)$ as the total cost of the variables
tested by $S$ on input $x$ until $f(x)$ is determined.
We say $x \sim p$ if $\Pr(x)=\prod_{i:x_i=1}p_i \prod_{i:x_i=0} (1-p_i)$.
Then $\cost_{c,p}(f, S) := \E_{x \sim p}[\cost_c(f,x,S)]$
is the expected evaluation cost of strategy $S$ when
$x$ is drawn according to the product distribution induced by $p$.

For fixed $n$, let $\A$ be
the set of adaptive strategies on $n$ variables 
and $\N$ be the set of non-adaptive strategies on $n$ variables.
We are interested in the quantities 
\begin{align*}
    \OPT_\A (f,c,p) := \min_{S \in \A} \cost_{c,p}(f,S)
    \mbox{ and }
    \OPT_\N (f,c,p) := \min_{S \in \N} \cost_{c,p}(f,S).
\end{align*}
We will omit $c$ and $p$ from the notation when
the costs and probabilities are clear from context.
A {\em (Boolean) read-once formula} is a tree, each of whose internal nodes are labeled either $\vee$ or $\wedge$. The internal nodes have two or more children. 
Each leaf is labeled with a Boolean variable $x_i \in \{x_1, \ldots, x_n\}$. The formula computes a Boolean function in the usual way.\footnote{Some definitions of a read-once formula allow negations in the internal nodes of the formula. By DeMorgan's laws, these negations can be ``pushed'' into the leaves of the formula,
resulting in a formula whose internal nodes are $\vee$ and $\wedge$, such that each variable $x_i$ appears in at most one leaf.}
A (Boolean) {\em DNF formula} is a formula of the form $T_1 \vee T_2 \vee \ldots T_m$ for some $m \geq 1$, such that each {\em term} $T_i$ is the conjunction ($\wedge$) of literals. A literal is a variable $x_i$ or a negated variable $\neg x_i$. The DNF formula is read-once if distinct terms contain disjoint sets of variables, without negations.
Read-once DNF formulas whose terms all contain the same number $w$ of literals are sometimes known as tribes formulas of width $w$ (cf.~\cite{ODonnellBook}). 

\section{Warm Up: Adaptivity Gaps for Read-Once DNFs}
\label{sec:read-once}

Let $f:\{0,1\}^n \rightarrow \{0,1\}$ be a read-once DNF formula.
Boros and \"{U}ny\"{u}lurt~\cite{BorosUnyulurt00} 
showed that the following approach gives an optimal adaptive
strategy for evaluating $f$
(this has been rediscovered in later papers~\cite{DBLP:journals/ai/GreinerHJM06,kowshik2011information,KaplanKM05}).
Let $f=T_1 \vee T_2 \ldots \vee T_k$ be a DNF formula with $k$ terms.  For each term $T_j$,
let $\ell(j)$ be the number of variables in term $T_j$.
Order the variables of $T_j$ as $x_{j_1}, x_{j_2}, \ldots x_{j_{\ell(j)}}$ in  non-decreasing order of the ratio $c_i /(1-p_i)$, i.e., so that 
$c_{j_1} /(1-p_{j_1}) \leq c_{j_2} /(1-p_{j_2}) \leq \ldots \leq c_{j_{\ell(j)}} /(1-p_{j_{\ell(j)}})$.
For evaluating the single term $T_j$, an optimal strategy
tests the variables in $T_j$ sequentially,
in the order $x_{j_1}, x_{j_2}, \ldots x_{j_{\ell(j)}}$, until a variable is found to be false, or until all variables are tested and found to be true.

Denote the probability of the term evaluating to true as $P(T_j) = \prod_{i=1}^{\ell(j)} p_{j_i}$ and the expected cost of this evaluation of the term
as $$C(T_j) = \sum_{i=1}^{\ell(j)} (\sum_{k=1}^i c_{j_k}\prod_{r=1}^{i-1} p_{j_r}).$$
An optimal algorithm for evaluating $f$ applies the above strategy sequentially to the terms $T$ of $f$, in non-decreasing order of
the ratio $C(T)/P(T)$, until either some term is found to be satisfied by $x$,
so $f(x)=1$,
or all terms have been evaluated and found to be falsified by $x$,
so $f(x) =0$. 
We will use this optimal adaptive strategy
in the remainder of the section.

In what follows, we will frequently describe non-adaptive strategies as performing the $n$ possible tests in a particular order.  We mean by this that the permutation representing the strategy lists the tests in this order.
The testing stops when the value of $f$ can be determined.


\subsection{Unit Costs and the Uniform Distribution}


\begin{algorithm}
    \DontPrintSemicolon
    \caption{Evaluating a read-once DNF where each variable has unit cost and uniform distribution.}\label{alg:algorithm}
    \SetKwInOut{Input}{Input}
    \SetKwInOut{Output}{Output}

    \Input{$n > 0$, read-once DNF $f:\{0,1\}^n \rightarrow \{0,1\}$ with $m$ terms}
    \Output{$\pi$ \tcp*{$O(\log n)$-approximation non-adaptive strategy for $f$}}
    $\pi \gets []$ \tcp*{empty list}
    \For{$i=1$ \KwTo $m$}{
        \eIf(\tcp*[f]{$T_i$ is $i$th shortest term in $f$}){$|T_i| \leq 2 \log n$}{$\pi \gets \pi \text{ } + $ all variables in $T_i$}{$\pi \gets \pi \text{ } + $ first $2 \log n$ variables in $T_i$}
    }
    $\pi \gets \pi \text{ } + $ remaining variables not in $\pi$
\end{algorithm}

We begin by showing that
the adaptivity gap for read-once DNFs, in the case of 
unit costs and the uniform distribution, 
is at most $O(\log n)$. 

\begin{theorem}\label{thm:rodnf-uu-upper}
Let 
$f:\{0,1\}^n \rightarrow \{0,1\}$ be a read-once DNF formula.
For unit costs and the uniform distribution,
there is a non-adaptive strategy $S$ such that 
$\cost(f,S) \leq O(\log n) \cdot \OPT_{\A}(f)$.
\end{theorem}

\begin{proof}[Proof Sketch]
Using the characterization of the optimal
adaptive strategy due to Boros and 
\"{U}ny\"{u}lurt~\cite{BorosUnyulurt00},
we show that Algorithm~\ref{alg:algorithm} gives
a non-adaptive strategy that has expected cost
at most $O(\log n)$ times
the optimal adaptive strategy.
The algorithm crucially relies on the observation
that the optimal adaptive algorithm tests 
terms in non-decreasing order of length for unit
costs and the uniform distribution.
To see this, observe $C(T)/P(T)$ is non-decreasing
when terms are ordered by length in this setting.
For terms with length at most $2 \log n$,
we can test every variable without paying more than
$O(\log n)$ times the optimal adaptive strategy.
For terms with length greater than $2 \log n$,
we can test $2 \log n$ variables and only need
to continue testing with probability $1/n^2$.
\end{proof}

We complement Theorem \ref{thm:rodnf-uu-upper}
with a matching lower bound.
We prove the theorem by exhibiting a read-once
DNF with $\sqrt{n}$ identical terms.
We upper bound the optimal adaptive strategy
and argue any non-adaptive strategy has to make
$\log n$ tests per term to verify $f(x)=0$ which
occurs with constant probability.

\begin{theorem}\label{thm:rodnf-uu-lower}
Let $f:\{0,1\}^n \rightarrow \{0,1\}$
be a read-once DNF formula.
For unit costs and the uniform distribution,
$\OPT_\N(f) \geq \Omega(\log n) \cdot \OPT_\A(f)$.
\end{theorem}

\subsection{Unit Costs and Arbitrary Probabilities}

We give an upper bound of the adaptivity gap for read-once DNF formulas
with unit costs and arbitrary probabilities 
in the special case where all terms
have the same number of variables.
This is known as a tribes formula \cite{ODonnellBook}.
Let the number of terms be $m$.
We now describe two non-adaptive strategies which yield
a $n/m$-approximation and a $m$-approximation, respectively.
Then, by choosing the non-adaptive strategy based on the
the number of terms $m$, we are
guaranteed a $\min \{n/m, m\} \leq O(\sqrt{n})$-approximation.

\begin{lemma}\label{lemma:n/m-approx}
    Consider a read-once DNF $f:\{0,1\}^n \rightarrow \{0,1\}$
    where each term has the same number of variables.
    For unit costs and arbitrary probabilities,
    there is a non-adaptive strategy $S \in \N$ such that $\cost(f,S) \leq n/m \cdot \OPT_\A(f)$.
\end{lemma}

\begin{proof}[Proof of Lemma \ref{lemma:n/m-approx}]
    Consider a random input $x$ and the optimal adaptive strategy described at the start of this section.
    If $f(x)=0$, the optimal adaptive strategy must
    certify that each term is 0 which requires at least
    $m$ tests.
    Since any non-adaptive strategy will make at most $n$
    tests, the ratio between the cost incurred on $x$ by a non-adaptive strategy, and by the optimal adaptive strategy, is at most $n/m$.
    Otherwise, if $f(x)=1$, the optimal adaptive strategy
    will certify that a term is true after testing some number
    of false terms. 
    Now consider the non-adaptive version of this optimal
    adaptive strategy which
    tests terms in the same fixed order but must test all variables in
    a term before proceeding to the next term.
    For each false term that the optimal adaptive strategy
    tests, the non-adaptive strategy will test every variable
    for a total of $n/m$ tests.
    Since the optimal adaptive strategy must make at least
    one test per false term,
    the ratio between the cost incurred on $x$ by the non-adaptive strategy, and the cost incurred by the optimal strategy, is at most $n/m$.  Since the ratio $n/m$ holds for all $x$, the lemma follows.
\end{proof}

\begin{lemma}\label{lemma:m-approx}
    Consider a read-once DNF $f:\{0,1\}^n \rightarrow \{0,1\}$
    where each term has the same number of variables.
    For unit costs and arbitrary probabilities,
    there is a non-adaptive strategy $S \in \N$ with
    expected cost $\cost(f,S) \leq m \cdot \OPT_\A(f)$.
\end{lemma}

\begin{proof}[Proof of Lemma \ref{lemma:m-approx}]
    Fix a random input $x$.
    If $f(x)=0$, the optimal adaptive strategy certifies
    that every term is false.
    Let $C_i$ be the number of tests it makes until
    finding a false variable on the $i$th term.
    Consider the non-adaptive ``round-robin''
    strategy which progresses in rounds,
    making one test in each term per round.
    Within a term, the non-adaptive strategy tests
    variables in the same fixed order as the optimal adaptive strategy.
    Then the cost of the non-adaptive strategy is
    $m \cdot \max_i C_i$ whereas the cost of the optimal adaptive
    strategy is $\sum_{i=1}^m C_i$.
    It follows that the adaptivity gap is at most $m$.
    Otherwise, if $f(x)=1$, the optimal adaptive strategy
    must certify that a term is true by making
    at least $n/m$ tests.
    Any non-adaptive strategy will make at most $n$ tests
    so the adaptivity gap is at most $m$.
\end{proof}

Together, the $O(n/m)$- and $O(m)$-approximations
imply the following result.

\begin{theorem}\label{thm:rodnf-ucap-upper}
    Let $f:\{0,1\}^n \rightarrow \{0,1\}$
    be a read-once DNF formula where each term has the same number of variables.
    For unit costs and arbitrary probabilities,
    there is a non-adaptive strategy $S \in \N$
    with $\cost(f,S) \leq O(\sqrt{n}) \cdot \OPT_\A(f)$.
\end{theorem}

We complement Theorem \ref{thm:rodnf-ucap-upper} with a matching lower bound.
We prove the theorem by exhibiting a read-once DNF
with $2\sqrt{n}$ identical terms.
By making one special variable in each term have a low
probability of being true and arguing it must always
be tested first, the non-adaptive strategy has to 
search at random for which special variable is true
when every other special variable is false which
happens with constant probability.

\begin{theorem}\label{thm:rodnf-ucap-lower}
Let $f:\{0,1\}^n \rightarrow \{0,1\}$
be a read-once DNF formula.
For unit costs and arbitrary probabilities,
$\OPT_\N(f) \geq \Omega(\sqrt{n}) \cdot \OPT_\A(f)$.
\end{theorem}

\subsection{Arbitrary Costs and the Uniform Distribution}

We prove Theorem \ref{thm:lambert} by exhibiting a
read-once DNF with $2^\ell$ terms each of length $\ell$.
Within each term, the cost of each variable increases
geometrically with a ratio of 2.
The challenge is choosing $\ell$ so that $2^\ell \ell = n$.
We accomplish this by using a modified Lambert W function 
\cite{bronstein2008algebraic}
which is how we calculate $n_\epsilon$.

\begin{theorem}\label{thm:lambert}
    For all $\epsilon > 0$, 
    there exists $n_{\epsilon} > 0$ such that the
    following holds for all read-once DNF formulas $f:\{0,1\}^n \rightarrow \{0,1\}$ where $n > n_{\epsilon}$: There exists a cost assignment such that for the uniform distribution,
    $\OPT_\N(f) \geq \Omega(n^{1-\epsilon}/\log n)
    \cdot \OPT_\A(f)$.
\end{theorem}


\section{Main Result:
Read-Once Formulas}\label{sec:formula}

\begin{theorem}\label{thm:read-once-gap}
    Fix $\epsilon>0$.
    There is a read-once formula
    $f:\{0,1\}^n \rightarrow \{0,1\}$, such that
    for unit costs and $p_i=\frac{1+\epsilon}{2}$ for all $i \in [n]$,
    $\OPT_\N(f) \geq 
    \Omega\left(\epsilon^3 n^{1-2\epsilon/\log 2}\right)
    \cdot \OPT_\A(f)$.
\end{theorem}

Before we prove Theorem \ref{thm:read-once-gap},
we describe the read-once formula $f$ and
present the technical lemmas we use in the proof.
Without loss of generality, assume
$n=2 \cdot 2^d-2$ for some positive
integer $d$.
We define the function $f(x)$ on inputs $x\in \{0,1\}^n$
in terms of a binary tree with depth $d$. 
The edges of the tree are numbered $1$ through $n$,
and variable $x_i$  corresponds to edge $i$.
Each variable $x_i$ has a 
$\frac{1+\epsilon}{2}$ probability
of being true.
Say that a leaf of the tree is ``alive'' if
$x_i=1$ for all edges $i$ on the path from the root to the leaf.
We define $f(x)=1$ if and only if at 
least one leaf of the tree is alive.
A strategy for evaluating $f$ will continue testing until it can certify that there is at least one alive leaf, or that no alive leaf exists.
\begin{figure}[ht]
    \centering
    \includegraphics[scale=.16]{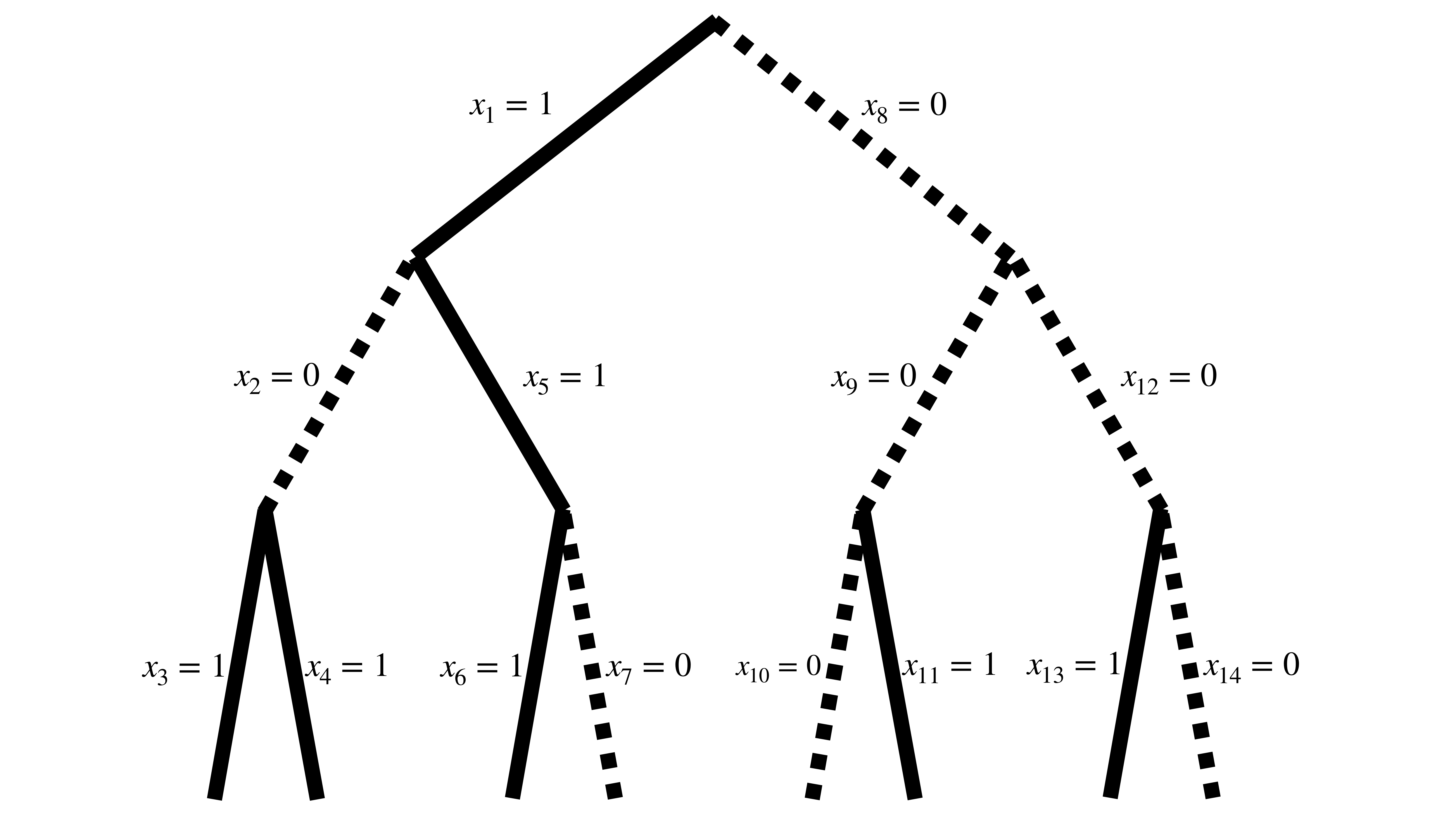}
    \caption{The binary tree corresponding
    to the read-once formula we construct
    when $n=14$. In particular,
    $f(x)=(x_1\wedge((x_2\wedge(x_3\vee x_4))\vee(x_5\wedge(x_6\vee x_7))))\vee (x_8\wedge((x_9\wedge(x_{10}\vee x_{11}))\vee(x_{12}\wedge(x_{13}\vee x_{14}))))$.
    Notice that $f(x)=1$ for this $x$
    because the third leaf from the left is alive
    (all its ancestors are true).}
    \label{fig:branching_tree}
\end{figure}

Now consider the multi-graph that is produced from the tree
by merging all leaves into a single node.
The function $f$ is the $st$-connectivity function of this multi-graph.
It is easy to show, by induction on the depth of the tree, that the multi-graph is two-terminal series-parallel, with $s$ the root, and $t$ the node produced by merging the leaves of the tree.
Thus $f$ is computed by a read-once formula.  

We refer to the edges of the tree that join a leaf to its parent as leaf edges, and the other edges as internal edges.
 We say that a non-adaptive strategy $S$ is leaf-last if it first tests
all non-leaf edges of the tree, and then tests the leaf edges.  

In the proof of Theorem \ref{thm:read-once-gap},
we consider an alternative cost assignment where we pay unit costs for the tests on leaf edges, as usual, but tests on internal edges are free. 
The expected cost of a strategy under the usual unit cost assignment
is clearly lower bounded by its expected cost when internal edges are free.
Note that when internal edges are free,
there is no disadvantage in performing all the tests on internal edges first, so
there is an optimal non-adaptive strategy which is leaf-last
in the sense that all the leaf edges appear last.
Our first technical lemma describes a property of a leaf-last strategy.
We defer the proof of this lemma, and of the ones that follow, to the end of this section.
In all of the lemma statements, we assume $f$ is as just described, and expected costs are with respect to unit costs and test probabilities $p_i = \frac{1+\epsilon}{2}$.  We use $L$ to denote the number of leaves in the tree.


\begin{lemma}\label{lemma:nonincreasingprob}
There exists a leaf-last non-adaptive strategy $S$ for evaluating $f$
which,
conditioned on the event that there is at least one alive leaf,
has minimum expected cost when internal edges are free relative to all non-adaptive strategies. Further, for any such $S$ and any $\ell \in [L-1]$, 
   %
    conditioned on the existence of at least one alive leaf, the probability that $S$
    first finds an alive leaf on the $\ell$th leaf test
    is at least the probability $S$
    first finds an alive leaf on the
    $(\ell+1)$st leaf test. 
\end{lemma}

The next lemma gives us an
inequality that we will use to lower bound
the cost of the optimal non-adaptive strategy.

\begin{lemma}\label{lemma:earthmover}
    Let $L$ be a positive integer
    and $p_1 \geq p_2 \geq \ldots \geq p_L$
    be non-negative real numbers.
    Now let $p\geq p_1$ and define
    $L' = \lfloor \sum_{\ell=1}^L p_\ell/p \rfloor$.
    Then
    $
        \sum_{\ell=1}^L \ell p_\ell
        \geq \sum_{\ell=1}^{L'} \ell p.
    $
\end{lemma}

Our analysis depends on there being at least constant probability that $f(x)=1$, or equivalently, that there is at least one alive leaf.
%
The next lemma assures us that this is indeed the case.  The proof of the lemma depends on our choice of having each $p_i$ be slightly larger than 1/2; it would not hold otherwise.  

\begin{lemma}\label{lemma:not2many}
With probability at least $\epsilon$,
there is at least one alive leaf in the binary tree representing $f$. 
\end{lemma}

With these key lemmas in hand, we prove
Theorem \ref{thm:read-once-gap}.

\begin{proof}[Proof of Theorem \ref{thm:read-once-gap}]
We will show that the adaptivity gap is large.
Intuitively, we rely on the fact that
if there is at least one alive leaf, 
then an adaptive strategy can find an alive leaf cheaply, 
by beginning at the root of the tree
and moving downward only along edges that are alive.
In contrast, a non-adaptive strategy cannot stop searching along
``dead'' branches.
However, it is not immediately clear that the cost of the
non-adaptive strategy is high because there are conditional
dependencies between the probabilities that two leaves with the same
ancestor(s) are alive.
To prove the desired result, we need to show that,
despite these dependencies, the optimal non-adaptive strategy
must have high expected cost.

We begin by showing that the expected cost
of any non-adaptive strategy is at least
$\frac{\epsilon^2}{16} n^{1-\frac{\epsilon}{\log 2}}$.
We want to lower bound the expected cost
of the optimal strategy $\OPT_\N(f)$:
\begin{align}\label{eq:dnf-leaf-cost}
    \min_{S \in \N} \E_{x}[\cost(f, x, S)]\geq
    \min_{S} \E[\cost^L(f, x, S)] =
    \min_{S'}
    \E[\cost^L(f, x, S')]
\end{align}
where $\cost^L(f, x, S)$ is the number of
leaf tests $S$ makes on $x$ until
$f(x)$ is determined, and
$S'$ is a leaf-last strategy.
Then
\begin{align*}
    (\ref{eq:dnf-leaf-cost}) 
    &= \min_{S'} \left(
    \sum_{x:f(x)=1} \Pr(x) \cdot \cost^L(f, x, S')
    + \sum_{x:f(x)=0} \Pr(x) \cdot \cost^L(f, x, S')
    \right) \\
    &\geq \min_{S'}
    \sum_{x:f(x)=1} \Pr(x) \cdot \cost^L(f, x, S') \\
    &= \min_{S'}
    \sum_{\ell=1}^L \ell \Pr(\textrm{$S'$ first finds
    alive leaf on $\ell$th leaf test})
\end{align*}
where $L=2^d$ is the number of 
leaves in the binary tree.
Initially, all leaves have a $\left(\frac{1+\epsilon}{2}\right)^d$
probability of being alive
where $d = \log_2((n+2)/2)$.
By Lemma \ref{lemma:nonincreasingprob},
the probability that the next leaf is alive 
cannot increase as the optimal
non-adaptive strategy $S^*$ performs its test.
Set $$p=\left(\frac{1+\epsilon}{2}\right)^d$$ and 
$p_\ell = \Pr(\textrm{$S^*$ first finds
alive leaf on $\ell$th test})$.
Therefore, Lemma \ref{lemma:earthmover} with $p$ and $p_\ell$
tells us that 
\begin{align}\label{eq:binarynonlower}
    \sum_{\ell=1}^L \ell p_\ell
    \geq \sum_{\ell=1}^{L'} \ell 
    \left(\frac{1+\epsilon}{2}\right)^d
    \geq \left(\frac{1+\epsilon}{2}\right)^d
    \frac{\epsilon^2}
    {8\left(\frac{1+\epsilon}{2}\right)^{2d}}
    \geq \frac{\epsilon^2}
    {8 (\frac{n+2}{2})^{\log_2(\frac{1+\epsilon}{2})}}
    \geq \frac{\epsilon^2}{16} 
    n^{1-\frac{\epsilon}{\log 2}}
\end{align}
where we use the inequality that
$L' \geq \epsilon/(2\left(\frac{1+\epsilon}{2}\right)^d)$.
To see this,
recall that $\sum_{\ell=1}^L p_\ell/p \ge L'$
and, since the right-hand side is greater than 1,
$2L' \ge \sum_{\ell=1}^L p_\ell/p$.
By Lemma \ref{lemma:not2many},
$\sum_{\ell=1}^L p_\ell \ge \epsilon$ so
$L' \geq \epsilon/(2\left(\frac{1+\epsilon}{2}\right)^d)$.
The last inequality in Equation (\ref{eq:binarynonlower}) follows from
$\log_2(\frac{1+\epsilon}{2}) \leq \frac{\epsilon}{\log 2}-1$
which can be shown by comparing the $y$-intercepts
and derivatives for $\epsilon > 0$.

Next, we show that the expected cost
of the adaptive strategy is at most
$(n+1)^{\frac{\epsilon}{\log 2}}/\epsilon$.
Consider an adaptive strategy which starts by querying the
two edges of the root
and recurses as follows:
if an edge is alive it queries
its two child edges and otherwise stops.
Observe that this simple depth-first search
adaptive strategy will make at most two tests
for every alive edge in the binary tree.
Therefore the expected number of tests an adaptive
strategy must make is at most twice
the expected number of alive edges.
By the branching process analysis in the proof of Lemma
\ref{lemma:not2many},
twice the expected number of alive edges is
\begin{align}\label{eq:binaryadaptupper}
    2 \sum_{i=0}^d (1+\epsilon)^i
    \leq 2 \sum_{i=0}^{\log_2 n} (1+\epsilon)^i
    = 2 \frac{(1+\epsilon)^{\log_2 (n)+1}-1}
    {(1+\epsilon)-1}
    \leq 4 \frac{n^{\log_2(1+\epsilon)}}{\epsilon}
    \leq 4 \frac{n^{\frac{\epsilon}{\log 2}}}
    {\epsilon}
\end{align}
where the last inequality follows
from $\log_2(1+\epsilon) \leq \frac{\epsilon}{\log 2}$
which we can see by comparing the $y$-intercepts
and slopes for $\epsilon>0$.
Then Theorem \ref{thm:read-once-gap} follows
from Equations (\ref{eq:binarynonlower})
and (\ref{eq:binaryadaptupper}).
\end{proof}

\begin{proof}[Proof of Lemma \ref{lemma:nonincreasingprob}]
    A leaf-last non-adaptive strategy $S$ satisfying the given conditional optimality property clearly exists, because any non-adaptive strategy can be made leaf-last by moving the leaf tests to the end without affecting its cost when internal edges are free.
    Define $p_\ell(S)$ as the probability that
    $S$ finds an alive leaf for the first
    time on leaf test $\ell$.
    We may write the expected number of
    leaf tests of $S$ as
    $\sum_{\ell=1}^L \ell p_\ell(S)$.
    
    Now suppose for contradiction that
    there is some $\ell'$ such that
    $p_{\ell'}(S) < p_{\ell'+1}(S)$.
    Let $S'$ be $S$
    but with the $\ell'$th and $(\ell'+1)$th
    tests swapped.
    We will show that the expected
    number of leaf tests made by $S'$
    is strictly lower than the expected
    number of leaf tests made by $S$.
    Observe that
    \begin{align*}
    \sum_{\ell=1}^L \ell p_\ell(S)
    - \sum_{\ell=1}^L \ell p_\ell(S')
    = \ell' (p_{\ell'}(S)-p_{\ell'}(S'))
    + (\ell'+1)(p_{\ell'+1}(S)-p_{\ell'+1}(S')).
    \end{align*}
    Notice that
    $p_{\ell'}(S) < p_{\ell'+1}(S) \leq p_{\ell'}(S')$
    where the first inequality
    follows by assumption
    and the second inequality follows
    because moving a test on a particular leaf edge to appear
    earlier in the permutation
    can only increase the probability
    that its leaf is the first alive leaf found.
    In addition, since
    the combined probability we first find an alive
    leaf in either the $\ell'$th or 
    $(\ell'+1$)th test is the same
    in either order of tests,
    $p_{\ell'}(S)+p_{\ell'+1}(S)
    = p_{\ell'}(S')+p_{\ell'+1}(S')$.
    Together,
    we have that
    $-(p_{\ell'}(S)-p_{\ell'}(S'))=
    p_{\ell'+1}(S)-p_{\ell'+1}(S')
    > 0$.
    Therefore
    $\sum_{\ell=1}^L \ell p_\ell(S)
    - \sum_{\ell=1}^L \ell p_\ell(S')=
    p_{\ell'+1}(S)-p_{\ell'+1}(S') > 0$
    and $S'$ makes fewer leaf tests
    in expectation even though $S$
    was optimal by assumption.
    A contradiction!
\end{proof}

\begin{proof}[Proof of Lemma \ref{lemma:earthmover}]
Define $\delta_\ell = p - p_\ell \geq 0$
for $\ell \in [L']$.
Observe that $L' \leq L$ since
$p\geq p_\ell$ for all $\ell \in [L']$.
Then
\begin{align}\label{eq:delta_ell}
    \sum_{\ell=1}^L \ell p_\ell =
    \sum_{\ell=1}^{L'} \ell (p-\delta_\ell) +
    \sum_{\ell=L'+1}^L \ell p_\ell =
    \sum_{\ell=1}^{L'} \ell p
    - \sum_{\ell=1}^{L'} \ell \delta_\ell +
    \sum_{\ell=L'+1}^L \ell p_\ell.
\end{align}
Now all that remains to be shown is that
the sum of the last two terms in
Equation (\ref{eq:delta_ell}) is non-negative.
Notice that
\begin{align*}
    \sum_{\ell=1}^{L'} p_\ell
    + \sum_{\ell=L'+1}^{L} p_\ell \ge
    \sum_{\ell=1}^{L'} p
    \implies
    \sum_{\ell=L'+1}^L p_\ell \ge
    \sum_{\ell=1}^{L'} (p-p_\ell) =
    \sum_{\ell=1}^{L'} \delta_\ell.
\end{align*}
Then
\begin{align*}
    \sum_{\ell=1}^{L'} \ell \delta_\ell \leq
    L' \sum_{\ell=1}^{L'} \delta_\ell \leq
    L' \sum_{\ell=L'+1}^L p_\ell \leq
    \sum_{\ell=L'+1}^L \ell p_\ell.
\end{align*}
\end{proof}

\begin{proof}[Proof of Lemma \ref{lemma:not2many}]
Let $Z_i$ denote the number of alive
edges at level $i$.
Then the statement of Lemma \ref{lemma:not2many} becomes 
$\Pr(0< Z_d) \geq \epsilon$.
Using standard results from the study of branching
processes, we know that
\begin{align*}
  \E[Z_d] = \mu^d
  \hspace{1em} \textnormal{ and } \hspace{1em}
  \Var{Z_d} = \left(\mu^{2d} - \mu^{d} \right)
  \frac{\sigma^2}{\mu (\mu-1)}
\end{align*}
where
$\mu$ is the expectation
and $\sigma$ is the variance
of the number of alive ``children'' from a single alive edge
(see e.g., p. 6 in Harris \cite{harris1963theory}).
In our construction,
\begin{align*}
   \mu = 2\left(\frac{1+\epsilon}{2}\right)
    = (1+\epsilon)
   \hspace{1em} \textnormal{ and } \hspace{1em}
   \sigma^2 
   = 2\left(\frac{1+\epsilon}{2}\right)
   \left(1-\frac{1+\epsilon}{2}\right)
   =\frac{1-\epsilon^2}{2}
\end{align*}
since the children of one edge follow the binomial
distribution.
Then $\E[Z_d] = (1+\epsilon)^d$
and
\begin{align*}
    \Var{Z_d} = \left( (1+\epsilon)^{2d}
    - (1+\epsilon)^d \right)
    \frac{1-\epsilon^2}{2(1+\epsilon)\epsilon}
    \leq \frac{1}{2\epsilon}
    \left((1+\epsilon)^{d}\right)^2
    = \frac{1}{2\epsilon} \E[Z_d]^2.
\end{align*}

We will now use Cantelli's inequality
(see page 46 in Boucheron et al. \cite{boucheron2013concentration})
to show that $\Pr(Z_d > 0) \geq \epsilon$.
Cantelli's tells us that 
$\Pr(X - \E[X] \geq \lambda) 
\leq \frac{\Var{X}}{\Var{X}+\lambda^2}$
for any real-valued random variable $X$ and 
$\lambda > 0$.
Choose $X= -Z_d$ and $\lambda = \E[Z_d]$.
Then
\begin{align*}
    \Pr(Z_d \leq \E[Z_d] - \E[Z_d])
    \leq \frac{\Var{Z_d}}{\Var{Z_d}+\E[Z_d]^2}
\end{align*}
and, by taking the complement,
\begin{align}\label{eq:applied-cantelli}
    \Pr(Z_d > 0) \geq 
    \frac{\E[Z_d]^2}{\Var{Z_d}+\E[Z_d]^2}
    \geq \frac{\E[Z_d]^2}
    {\frac{1}{2\epsilon}\E[Z_d]^2
    +\E[Z_d]^2} = \frac{2\epsilon}{1+2\epsilon}
    \geq \epsilon
\end{align}
for $0 < \epsilon \leq 1/2$.
Then Lemma \ref{lemma:not2many}
follows from Equation (\ref{eq:applied-cantelli}).
\end{proof}

\section{DNF Formulas}
\label{sec:dnf}

We will show near-linear and linear in $n$ lower 
bounds for DNF formulas under the uniform distribution with
unit and arbitrary costs, respectively.
Since the function we exhibit has linear terms,
the lower bounds also apply to the class of linear-size DNF formulas.

\begin{theorem}\label{thm:nlognlower}
Let $f:\{0,1\}^n \rightarrow \{0,1\}$
be a DNF formula.
For unit costs and the uniform distribution,
$\OPT_\N(f) \geq \Omega(n/\log n) \cdot \OPT_\A(f)$.
\end{theorem}

\begin{proof}
Without loss of generality,
assume $n=2^d + d$ for some positive integer $d$.
Consider the address function $f$ with
$2^d$ terms which each consist of
$d$ shared variables appearing in all terms, and a
single dedicated variable appearing only in that term.
We may write $f = T_0\vee T_1\vee\cdots\vee T_{2^d-1}$
where $T_i$ consists of the shared variables
negated according to the binary representation of $i$
and the single dedicated variable.

By testing the $d$ shared variables, the optimal
adaptive strategy can learn which single term
is unresolved and test the corresponding dedicated
variable in a total of $d+1$ tests.
In contrast, any non-adaptive strategy has to search
for the unresolved dedicated test at random which
gives expected $2^d/2$ cost.
(We can ensure the non-adaptive strategy tests
the shared variables first by making them free which can
only decrease the expected cost.)
It follows that the adaptivity gap
is $\Omega(2^d/d) = \Omega(n/\log n)$.
\end{proof}

%
%

We can easily modify the address function
in the proof of Theorem \ref{thm:nlognlower}
to prove an $\Omega(n)$ lower bound
for DNF formulas under the uniform
distribution and arbitrary costs.
In particular, make the cost of each shared variable $1/d$.
Then the adaptive strategy
pays $d\cdot(1/d) + 1=2$
while the non-adaptive strategy still
pays $\Omega(n)$.
The $O(n)$ upper bound comes from the
increasing cost strategy and analysis in
\cite{DBLP:conf/stoc/KaplanKM05}.

\begin{theorem}\label{thm:nlower}
Let $f:\{0,1\}^n \rightarrow \{0,1\}$
be a DNF formula.
For the uniform distribution,
$\OPT_\N(f) \geq \Theta(n) \cdot \OPT_{\A}(f)$.
\end{theorem}

\section{Conclusion and Open Problems}

We have shown bounds on the adaptivity gaps 
for the SBFE problem for well-studied classes of Boolean formulas.
Our proof of the lower bound for read-once formulas depended on having $p_i$'s that are slightly larger than 1/2 but we conjecture that a similar or better lower bound holds for the uniform distribution.  
We note that our lower bound for read-once formulas also applies to (linear-size) monotone DNF formulas, since the given read-once formula based on the binary tree has a DNF formula with one term per leaf.  Another open question is to prove a lower bound for monotone DNF formulas that matches our lower bound for general DNF formulas.  

A long-standing open problem is whether the SBFE problem for read-once formulas has a polynomial-time algorithm (cf.~\cite{DBLP:journals/ai/GreinerHJM06,Unluyurt04}).
The original problem only considered adaptive strategies, and it is also open whether there is a
polynomial-time (or pseudo polynomial-time) constant or $\log n$ approximation algorithm for such strategies.
Happach et al.~\cite{happach2021general} gave a pseudo polynomial-time approximation algorithm 
for the non-adaptive version of the problem, which outputs a non-adaptive strategy with expected cost within a constant factor of the optimal non-adaptive strategy.
Because of the large adaptivity gap for read-once formulas,
as shown in this paper, the result of Happach et al.
does not have any implications for the open question
of approximating the adaptive version of the SBFE
problem for read-once formulas.



\bibliography{miscrefs.bib}

\appendix

\section{Additional Proofs}\label{app:delayed}

\begin{proof}[Proof of Theorem \ref{thm:rodnf-uu-upper}]
Suppose $f$ is a read-once DNF formula.
We will prove that for unit costs and the uniform distribution,
there is a non-adaptive strategy $S$ such that
$\cost(f,S) \leq O(\log n) \cdot \OPT_{\A}(f)$.

Let $m$ be the number of terms in $f$.
Because each variable $x_i$ appears in at most one term,
we have that $m \leq n$.
As a warm-up, we begin by proving adaptivity gaps 
for two special cases of $f$.

\paragraph*{Case 1: All terms have at most $2\log n$ variables}
Under the uniform distribution and with unit costs,
the $p_i$ are all equal, and the $c_i$ are all equal.  Thus
in this case,
the optimal adaptive strategy described previously tests
terms in increasing order of length.
The adaptive strategy \textit{skips} in the sense that
if it finds a variable in a term that is false, 
it moves to the next term without testing the remaining variables
in the term.
Suppose we eliminate skipping from the optimal adaptive strategy,
making the strategy non-adaptive.
Since all terms have at most $2\log n$ variables,
this increases the testing cost for any given $x$ 
by a factor of at most $2\log n$.
Thus the cost of evaluating $f(x)$ for a fixed $x$ 
increases by a factor of at most $2 \log n$ from
an optimal adaptive strategy to a non-adaptive strategy,
leading to an adaptivity gap of at most $2 \log n$.

\paragraph*{Case 2: All terms have more than $2\log n$ variables}
Consider the following non-adaptive strategy that operates in two phases.  
In Phase 1, the strategy tests a fixed subset of
$2 \log n$ variables from each term, where the terms
are taken in increasing length order.
In Phase 2, it tests the remaining untested
variables in fixed arbitrary order.
Since each term has more than $2\log n$ variables, the value $f$ can only be determined in Phase 1 if a false variable is found in each term during that phase.

Say that an assignment $x$ is \textit{bad} if the value of $f$ cannot be determined in Phase 1, meaning that a false variable is not found in every term during the phase.
The probability that a random $x$ satisfies all the tested 
$2 \log n$ variables of a particular term is $1/n^2$.  
Then, by the union bound, the probability that $x$ is bad 
is at most $m/n^2 \leq n/n^2 = 1/n$.

Now let us focus on the \textit{good} (not bad) assignments $x$.
For each good $x$, our strategy must find a false variable in each term of $f$, 
which requires at least one test per term for
any adaptive or non-adaptive strategy.
The cost incurred by our non-adaptive strategy 
on a good $x$ is at most $2m\log n$,
since the strategy certifies that 
$f(x)=0$ by the end of Phase 1.
Therefore, the expected cost incurred by our
non-adaptive strategy $S$ is
\begin{align*}
    \cost(f, S) &\le \Pr(x \mbox{ good})
    \cdot \E[\cost(f,x,S) | x \mbox{ good}] \\
    &+\Pr(x \mbox{ bad})
    \cdot \E[\cost(f,x,S) | x \mbox{ bad}] \\
    &\leq 1 \cdot 2m\log n + \frac{1}{n} \cdot n
    \leq 3m \log n
\end{align*}
using the fact that $\E[\cost(f,x,S) | x \mbox{ bad}] \leq n$, since there are only $n$ tests, with unit costs.

The expected cost of any strategy,
including the optimal adaptive strategy, is at least
\begin{align*}
   \OPT_\A(f) &\geq
   \min_{S \in \A} \Pr(f(x)=0) \cdot \E[\cost(f,x,S)|f(x)=0]
   \geq P(f(x)=0) \cdot m \\
   &= (1-\Pr(f(x)=1)) \cdot m 
   \ge (1-\Pr(x \mbox{ bad})) \cdot m 
   \ge \left(1- \frac{1}{n} \right) \cdot m 
   \ge \frac{m}{2}
\end{align*}
for $n \geq 2$.
It follows that the adaptivity gap is at most $6 \log n$.

\paragraph*{Case 3: Everything else}
We now generalize the ideas in the above two cases.
Let $f$ be a read-once DNF that
does not fall into Case 1 or Case 2.
We can break this DNF into two smaller DNFs,
$f = f_1 \vee f_2$ where $f_1$ contains the terms of $f$
of length at most $2\log n$ and
$f_2$ contains the terms of $f$ of length greater than $2 \log n$.

Let $S$ be the non-adaptive strategy that first applies
the strategy in Case 1 to $f_1$ and then, if $f_1(x)=0$,
the strategy in Case 2 to $f_2$.
Since $S$ cannot stop testing until it determines the value of $f$,
in the case that $f_1(x)=0$, it will test all variables in $f_1$ and then proceed
to test variables $f_2$.

Let $S^*$ be the optimal adaptive strategy for evaluating read-once DNFs, described above.
We know $S^*$ will test terms in non-decreasing
order of length since all tests are equivalent.
So, like $S$, $S^*$ tests $f_1$ first and then, if $f_1(x)=0$,
it continues to $f_2$.
It follows that we can write
the expected cost of $S$ on $f$ as
\begin{align*}
    \E[\cost(f,x,S)] = \E[\cost(f_1,x,S_1)]
    + \Pr(f_1(x)=0) \cdot \E[\cost(f_2, x, S_2)|f_1(x)=0]
\end{align*}
where $S_1$ is the first stage of $S$, where $f_1$ is evaluated, and $S_2$ is the second stage of $S$, where $f_2$ is evaluated.
Notice that, by the independence of variables, 
$\E[\cost(f_2, x, S_2)|f_1(x)=0] = \E[\cost(f_2, x, S_2)]$.
We can similarly write the expected cost of $S^*$ on $f$.
Then the adaptivity gap is
\begin{align}\label{eq:everythingelse}
    \frac{\OPT_\N(f)}{\OPT_\A(f)}
    \le \frac{\E[\cost(f_1,x,S_1)]
    + \Pr(f_1(x)=0) \cdot \E[\cost(f_2, x, S_2)]}
    {\E[\cost(f_1,x,S^*_1)]
    + \Pr(f_1(x)=0) \cdot \E[\cost(f_2, x, S^*_2)]}
\end{align}
where $S^*_1$ is $S^*$ applied to $f_1$
and $S^*_2$ is $S^*$ beginning from the point when it starts
evaluating $f_2$.

Using the observation that $(a+b)/(c+d) \leq \max\{a/c, b/d\}$
for positive real numbers $a,b,c,d$,
we know that
\begin{align*}
    (\ref{eq:everythingelse}) \le
    \max \left\{
    \frac{\E[\cost(f_1,x,S_1)]}{\E[\cost(f_1,x,S^*_1)]},
    \frac{\E[\cost(f_2,x,S_2)]}{\E[\cost(f_2,x,S^*_2)]}
    \right\}
    = O(\log n)
\end{align*}
where the upper bound follows from the analysis of Cases 1 and 2.
\end{proof}

\begin{proof}[Proof of Theorem \ref{thm:rodnf-uu-lower}]
Suppose $f$ is a read-once DNF formula.
For unit costs and the uniform distribution,
we will show that
$\OPT_\N(f) \geq \Omega(\log n) \cdot \OPT_\A(f)$.

For ease of notation, assume $\sqrt{n}$ is an integer.
Consider a read-once DNF $f$ with $\sqrt{n}$ terms
where each term has $\sqrt{n}$ variables.
By examining the number of tests in each term,
we can write the optimal adaptive cost as
\begin{align*}
    \OPT_\A(f) \leq \sqrt{n}
    \sum_{i=1}^{\sqrt{n}} \frac{i}{2^i}
    \le \sqrt{n} \sum_{i=1}^{\infty} \frac{i}{2^i}
    = 2 \sqrt{n}.
\end{align*}
The key observation is that, within a term,
the adaptive strategy queries variables
in any order since each variable is equivalent
to any other.
Then the probability that the strategy queries
exactly $i \leq \sqrt{n}$ variables is $1/2^i$.

Next, we will lower bound the expected cost
of the optimal non-adaptive strategy
\begin{align*}
    \OPT_\N(f) &= \min_{S \in \N}
    \E_{x \sim \{0,1\}^n}[\cost(f,x,S)] \\
    &\ge \min_{S \in \N}
    \Pr(f(x)=0)
    \E[\cost(f,x,S)|f(x)=0]
\end{align*}
where $x \sim \{0,1\}^n$ indicates $x$ is
drawn from the uniform distribution.
First, we know $\Pr(f(x)=0) \geq .5$.
To see this, consider a random input $x \sim \{0,1\}^n$.
The probability that a particular term is true
is $1/2^{\sqrt{n}}$ so the probability that all terms
are false (i.e., $f(x)=0$) is 
\begin{align*}
    \left(1-\frac{1}{2^{\sqrt{n}}} \right)^{\sqrt{n}}
    = \left(
    \left(1-\frac{1}{2^{\sqrt{n}}} \right)^{2^{\sqrt{n}}}
    \right) ^{\sqrt{n}/2^{\sqrt{n}}}
    \geq \left(\frac{1}{2e}\right)^{\sqrt{n}/2^{\sqrt{n}}}
    \geq .5
\end{align*}
where the first inequality follows from
the loose lower bound that $(1-1/x)^x \ge 1/(2e)$
when $x \geq 2$
and the second inequality follows when $n \geq 8$.
Second, we know
\begin{align*}
    &\E[\cost(f,x,S)|f(x)=0] \\
    &\ge 
    \Pr(\textrm{one term needs $\Omega(\log n)$ tests}|f(x)=0)
    \cdot \frac{\log_4 n}{2} \cdot \frac{\sqrt{n}}{2}
\end{align*}
where we used the symmetry of the terms
to conclude that if any term needs $\Omega(\log n)$ tests
to evaluate it then any non-adaptive strategy
will have to spend $\Omega(\log n)$ on half
the terms in expectation.

All that remains is to
lower bound the probability one
term requires $\Omega(\log n)$ tests given $f(x)=0$.
Observe that this probability is
\begin{align*}
    &1-(1-\Pr(\textrm{a particular term needs 
    $\Omega(\log n)$ tests}
    |f(x)=0))^{\sqrt{n}} \\
    &\ge 1-\left(1-\frac{1}{\sqrt{n}}\right)^{\sqrt{n}}
    \ge 1- \frac{1}{e} \ge .63
\end{align*}
where we will now show the first inequality.
We can write the probability that a particular
term needs $\log_4(n)/2$ tests given $f(x)=0$ as
\begin{align*}
    \Pr &\left(x_1=1|f(x)=0 \right)
    \cdots
    \Pr \left(x_{\log_4(n)/2}=1|f(x)=0, x_1=\cdots =x_{\log_4(n)/2-1}=1 \right) \\ 
    &= \frac{2^{\sqrt{n}-1}-1}{2^{\sqrt{n}}-1}
    \cdots
    \frac{2^{\sqrt{n}-1-\log_4(n)/2}-1}
    {2^{\sqrt{n}-\log_4(n)/2}-1} \ge 
    \left(\frac{2^{\sqrt{n}-1-\log_4(n)/2}-1}
    {2^{\sqrt{n}-\log_4(n)/2}-1}\right)^{\log_4(n)/2} \\
    &\ge \left(\frac{1}{4} \right)^{\log_4(n)/2}
    = \frac{1}{\sqrt{n}}.
\end{align*}
For the first equality, we use the 
observation that conditioning
on $f(x)=0$ eliminates the possibility every variable
is true
so the probability of observing a true variable
is slightly smaller.
For the first inequality,
notice that $(2^{i-1}-1)/(2^{i}-1)$ is
monotone increasing in $i$.
For the second,
observe that $i \ge \sqrt{n} - \log_4(n)/2$
for our purposes
and so $(2^{i-1}-1)/(2^{i}-1) \geq 1/4$
when $n \geq 16$.
\end{proof}

\begin{proof}[Proof of Theorem \ref{thm:rodnf-ucap-lower}]
Suppose $f$ is a read-once DNF.
For unit costs and arbitrary probabilities,
we prove
$\OPT_\N(f) \geq \Omega(\sqrt{n}) \cdot \OPT_\A(f)$.

Consider the read-once DNF with $m=2\sqrt{n}$ identical terms
where each term has $\ell=\sqrt{n}/2$ variables.
In each term, let one variable have
$1/\ell$ probability of being true
and the remaining variables have a $(\ell/m)^{1/(\ell-1)}$
probability of being true.
Within a term,
the optimal adaptive strategy will
test the variable with the lowest probability of being true first.
Using this observation, we can write
\begin{align*}
    \OPT_\A(f) \leq 
    \left[
    \Pr(x_1=0) \cdot 1 + \Pr(x_1=1) \cdot \ell
    \right] \cdot m \\
    \le \left[
    (1-1/\ell) \cdot 1 + (1/\ell) \cdot \ell
    \right] \cdot m
    \leq 4 \sqrt{n}
\end{align*}
where $x_1$ is the first variable tested in
each term.
The first inequality follows by charging
the optimal adaptive strategy for all $\ell$
tests in the term if the first one is true.
The second inequality follows since
the variable with probability $1/\ell$
of being true
is tested first for $n \ge 18$
(i.e., $1/\ell < (\ell/m)^{1/(\ell-1)}$
for such $n$).

In order to lower bound the cost of the 
optimal non-adaptive strategy,
we will argue that there is a constant probability
of an event where the non-adaptive strategy
has to test  $\Omega(n)$ variables.
In particular,
\begin{align*}
    \OPT_\N(f) \geq &\min_{S \in \N}
    \Pr(\mbox{exactly one term is true}) \\ \cdot
    &\E[\cost(f,x,S)|\mbox{ exactly one term is true}].
\end{align*}

By the symmetry of the terms,
observe that $$\E[\cost(f,x,S)|\mbox{ exactly one term is true}]
\ge \sqrt{n}/2 \cdot \sqrt{n} = n/2.$$
That is, the optimal non-adaptive strategy
has to search blindly for the single true term
among all $2\sqrt{n}$ terms,
making $\sqrt{n}/2$ tests each for half
the terms in expectation.

All that remains is to show there is a constant
probability exactly one term is true.
The probability a particular term is true
is $(1/\ell)((\ell/m)^{1/(\ell-1)})^{(\ell-1)} = 1/m$.
Since all variables are independent,
the probability that exactly one
of the $m$ terms is true is
\begin{align*}
    &m \cdot \Pr(\textrm{a term is true})
    \cdot \Pr(\textrm{a term is false})^{m-1} \\
    =& m \cdot \frac{1}{m} \cdot \left(1-\frac{1}{m} \right)^{m-1}
    \geq \frac{1}{2e}^{(m-1)/m} \geq \frac{1}{2e}.
\end{align*}

It follows that $\OPT_\N(f) \geq \frac{1}{2e} \cdot \frac{n}{2}
= \Omega(n)$ so
the adaptivity gap is $\Omega(\sqrt{n})$.
\end{proof}

\begin{proof}[Proof of Theorem \ref{thm:lambert}]
Suppose $f$ is a read-once formula.
For arbitrary costs and the uniform distribution,
$\OPT_\N(f) \geq \Omega(n^{1-\epsilon}/\log n)
\cdot \OPT_\A(f)$.

Define $W(w) := w^{1-\epsilon} \log_2(w^{1-\epsilon})$
for positive real numbers 
$w$.\footnote{Notice that $W$ is similar to a
Lambert W function $e^y y$, after changing the base of the logarithm and
substituting $y=\log(w^{1-\epsilon})$\cite{bronstein2008algebraic}.}
We will choose $n_\epsilon$ in terms of the function $W$
so that $W(n) < n$ for $n \ge n_\epsilon$.
First, consider the first and second derivatives of $W$:
\begin{align*}
    W'(w) &= \frac{1-\epsilon}{w^\epsilon}
    \left( \log_2(w^{1-\epsilon}) + \frac{1}{\log 2}
    \right) \\
    W''(w) &= \frac{1-\epsilon}{w^{1+\epsilon}}
    \left[ -\epsilon \left( \log_2(w^{1-\epsilon}) + \frac{1}{\log 2}\right)
    + \frac{1-\epsilon}{\log 2}
    \right].
\end{align*}
For fixed $\epsilon > 0$,
observe that as $w$ goes to infinity,
$W(w) < w$, $W'(w) < 1$, and $W''(w) < 0$.
Therefore there is some point $n_\epsilon$ so that
for all $n \ge n_\epsilon$, the slope of $W$ is decreasing,
the slope of $W$ is less than the slope of $n$,
and $W(n)$ is less than $n$.
Equivalently, $n \ge W(n) = n^{1-\epsilon}\log_2(n^{1-\epsilon})$.
We will use this inequality when lower bounding the asymptotic 
behavior of the adaptivity gap.

For $n \ge n_\epsilon$ we construct the $n$-variable read-once DNF formula $f$ as follows.
First, let $r_n$ be a real number such that $n = n^{1-r_n}\log_2(n^{1-r_n})$.
We know that $r_n$ exists for all $n \geq 4$ 
by continuity since
$n^{1-0} \log_2(n^{1-0}) \geq n \geq n^{1-1} \log_2(n^{1-1})$.
Let $f$ be the read-once DNF formula with $m$ terms of length $\ell$, where
$\ell = \log_2(n^{1-r_n})$ and $m=2^{\ell}$. 
Thus the total number of variables in $f$ is $m\ell = n^{1-r_n}\log_2(n^{1-r_n})=n$ as desired.
We assume for simplicity that $\ell$ is an integer.
The bound holds by a similar proof without this assumption.
%
%

To obtain our lower bound on evaluating this formula, we consider expected evaluation cost with respect to the
uniform distribution and the following cost assignment: 
in each term, choose an arbitrary ordering of the variables and set the cost of testing
the $i$th variable in the term to be $2^{i-1}$.

Consider a particular term.
Recall the optimal adaptive strategy for evaluating a read-once DNF formula presented at the start of Section~\ref{sec:read-once}. 
Within a term, this optimal strategy tests the variables in non-decreasing cost order,
since each variable has the same probability of being true.
Since it performs tests within a term until
finding a false variable or certifying
the term is true,
we can upper bound the expected cost 
of this optimal adaptive strategy in evaluating $f$ as follows:
\begin{align*}
    \OPT_\A(f) \leq m \cdot \left[
    \frac{1}{2}\cdot (1) + \frac{1}{4}\cdot (1+2)
    + \ldots +
    \frac{1}{2^\ell}\cdot(1+\ldots+2^{\ell-1}) \right]
    \leq m \cdot \ell.
\end{align*}
In contrast, the optimal non-adaptive strategy
does not have the advantage of stopping tests in a term
when it finds a false variable.
We will lower bound the expected cost
of the optimal non-adaptive strategy
in the case that exactly one term is true.
By symmetry, any non-adaptive strategy
will have to randomly search for the
term and so pay $2^{\ell}$ for
half the terms in expectation.

All that remains is to show there is a constant
probability exactly one term is true.
The probability that a particular term is true 
is $1/2^{\ell}$ and so the probability
that exactly one term is true is
\begin{align*}
    m \cdot \frac{1}{2^\ell} \cdot \left( 1-\frac{1}{2^\ell}\right)^{m-1} 
    \geq \frac{m}{2^\ell} \cdot
    \left(\frac{1}{2e} \right)^{(m-1)/2^\ell}
    \geq \frac{1}{2e}
\end{align*}
where the last inequality follows since $m=2^\ell$.
Then the expected cost $\OPT_\N(f)$ of the optimal
non-adaptive strategy is at least
\begin{align*}
    \Pr(\mbox{exactly one term is true})
    \cdot 2^\ell \cdot \frac{m}{2}
    = \Omega(m \cdot 2^\ell)
    = \Omega(m \cdot n^{1-r_n})
    \ge \Omega(m \cdot n^{1-\epsilon})
\end{align*}
where we used that $2^\ell = n^{1-r_n}$
and $n^{1-r_n} \log_2(n^{1-r_n})
= n \geq n^{1-\epsilon} \log_2(n^{1-\epsilon})$
since $n \ge n_\epsilon$.
It follows that the adaptivity gap is 
$\Omega(n^{1-\epsilon}/\log n)$.
\end{proof}

\end{document}